\documentclass[aps,pra,onecolumn,notitlepage,superscriptaddress,showpacs,nofootinbib]{revtex4-1}
\usepackage{amsmath}
\usepackage{amssymb}
\usepackage{amsfonts}
\usepackage{enumerate}
\usepackage{mathrsfs}
\usepackage{graphicx}
\usepackage{epstopdf}
\usepackage{enumerate}
\usepackage{changepage}
\usepackage{bm}
\usepackage{amsthm}
\usepackage{float}
\usepackage{lipsum}
\usepackage{diagbox}

\newtheorem{theorem}{Theorem}
\newtheorem{definition}{Definition}

\newtheorem{lemma}{Lemma}
\newtheorem{proposition}{Proposition}
\newtheorem{conjecture}{Conjecture}
\newtheorem{example}{Example}
\newtheorem{corollary}{Corollary}

\newtheoremstyle{named}{}{}{\itshape}{}{\bfseries}{.}{.5em}{\thmnote{#3's }#1}
\theoremstyle{named}
\newtheorem*{namedtheorem}{Theorem}

\def\bcj{\begin{conjecture}}
	\def\ecj{\end{conjecture}}
\def\bcr{\begin{corollary}}
	\def\ecr{\end{corollary}}
\def\bd{\begin{definition}}
	\def\ed{\end{definition}}
\def\bea{\begin{eqnarray}}
	\def\eea{\end{eqnarray}}
\def\bem{\begin{enumerate}}
	\def\eem{\end{enumerate}}
\def\bex{\begin{example}}
	\def\eex{\end{example}}
\def\bim{\begin{itemize}}
	\def\eim{\end{itemize}}
\def\bl{\begin{lemma}}
	\def\el{\end{lemma}}
\def\bma{\begin{bmatrix}}
	\def\ema{\end{bmatrix}}
\def\bpf{\begin{proof}}
	\def\epf{\end{proof}}
\def\bpp{\begin{proposition}}
	\def\epp{\end{proposition}}
\def\bqu{\begin{question}}
	\def\equ{\end{question}}
\def\br{\begin{remark}}
	\def\er{\end{remark}}
\def\bt{\begin{theorem}}
	\def\et{\end{theorem}}

\def\squareforqed{\hbox{\rlap{$\sqcap$}$\sqcup$}}
\def\qed{\ifmmode\squareforqed\else{\unskip\nobreak\hfil
		\penalty50\hskip1em\null\nobreak\hfil\squareforqed
		\parfillskip=0pt\finalhyphendemerits=0\endgraf}\fi}
\def\endenv{\ifmmode\;\else{\unskip\nobreak\hfil
		\penalty50\hskip1em\null\nobreak\hfil\;
		\parfillskip=0pt\finalhyphendemerits=0\endgraf}\fi}
\def\Dbar{\leavevmode\lower.6ex\hbox to 0pt
	{\hskip-.23ex\accent"16\hss}D}
\makeatletter
\def\url@leostyle{%
	\@ifundefined{selectfont}{\def\UrlFont{\sf}}{\def\UrlFont{\small\ttfamily}}}
\makeatother
\def\bcj{\begin{conjecture}}
	\def\ecj{\end{conjecture}}
\def\bcr{\begin{corollary}}
	\def\ecr{\end{corollary}}
\def\bd{\begin{definition}}
	\def\ed{\end{definition}}
\def\bea{\begin{eqnarray}}
	\def\eea{\end{eqnarray}}
\def\bem{\begin{enumerate}}
	\def\eem{\end{enumerate}}
\def\bex{\begin{example}}
	\def\eex{\end{example}}
\def\bim{\begin{itemize}}
	\def\eim{\end{itemize}}
\def\bl{\begin{lemma}}
	\def\el{\end{lemma}}
\def\bpf{\begin{proof}}
	\def\epf{\end{proof}}
\def\bpp{\begin{proposition}}
	\def\epp{\end{proposition}}
\def\bqu{\begin{question}}
	\def\equ{\end{question}}
\def\br{\begin{remark}}
	\def\er{\end{remark}}
\def\bt{\begin{theorem}}
	\def\et{\end{theorem}}

\def\btb{\begin{tabular}}
	\def\etb{\end{tabular}}

	\newcommand{\nc}{\newcommand}
	
	\nc{\bbA}{\mathbb{A}} \nc{\bbB}{\mathbb{B}} \nc{\bbC}{\mathbb{C}}
	\nc{\bbD}{\mathbb{D}} \nc{\bbE}{\mathbb{E}} \nc{\bbF}{\mathbb{F}}
	\nc{\bbG}{\mathbb{G}} \nc{\bbH}{\mathbb{H}} \nc{\bbI}{\mathbb{I}}
	\nc{\bbJ}{\mathbb{J}} \nc{\bbK}{\mathbb{K}} \nc{\bbL}{\mathbb{L}}
	\nc{\bbM}{\mathbb{M}} \nc{\bbN}{\mathbb{N}} \nc{\bbO}{\mathbb{O}}
	\nc{\bbP}{\mathbb{P}} \nc{\bbQ}{\mathbb{Q}} \nc{\bbR}{\mathbb{R}}
	\nc{\bbS}{\mathbb{S}} \nc{\bbT}{\mathbb{T}} \nc{\bbU}{\mathbb{U}}
	\nc{\bbV}{\mathbb{V}} \nc{\bbW}{\mathbb{W}} \nc{\bbX}{\mathbb{X}}
	\nc{\bbZ}{\mathbb{Z}}
	
	\nc{\bA}{{\bf A}} \nc{\bB}{{\bf B}} \nc{\bC}{{\bf C}}
	\nc{\bD}{{\bf D}} \nc{\bE}{{\bf E}} \nc{\bF}{{\bf F}}
	\nc{\bG}{{\bf G}} \nc{\bH}{{\bf H}} \nc{\bI}{{\bf I}}
	\nc{\bJ}{{\bf J}} \nc{\bK}{{\bf K}} \nc{\bL}{{\bf L}}
	\nc{\bM}{{\bf M}} \nc{\bN}{{\bf N}} \nc{\bO}{{\bf O}}
	\nc{\bP}{{\bf P}} \nc{\bQ}{{\bf Q}} \nc{\bR}{{\bf R}}
	\nc{\bS}{{\bf S}} \nc{\bT}{{\bf T}} \nc{\bU}{{\bf U}}
	\nc{\bV}{{\bf V}} \nc{\bW}{{\bf W}} \nc{\bX}{{\bf X}}
	\nc{\ba}{{\bf a}} \nc{\be}{{\bf e}} \nc{\bu}{{\bf u}}
	\nc{\brr}{{\bf r}}
	
	\nc{\cA}{{\cal A}} \nc{\cB}{{\cal B}} \nc{\cC}{{\cal C}}
	\nc{\cD}{{\cal D}} \nc{\cE}{{\cal E}} \nc{\cF}{{\cal F}}
	\nc{\cG}{{\cal G}} \nc{\cH}{{\cal H}} \nc{\cI}{{\cal I}}
	\nc{\cJ}{{\cal J}} \nc{\cK}{{\cal K}} \nc{\cL}{{\cal L}}
	\nc{\cM}{{\cal M}} \nc{\cN}{{\cal N}} \nc{\cO}{{\cal O}}
	\nc{\cP}{{\cal P}} \nc{\cQ}{{\cal Q}} \nc{\cR}{{\cal R}}
	\nc{\cS}{{\cal S}} \nc{\cT}{{\cal T}} \nc{\cU}{{\cal U}}
	\nc{\cV}{{\cal V}} \nc{\cW}{{\cal W}} \nc{\cX}{{\cal X}}
	\nc{\cZ}{{\cal Z}}
	
	\nc{\hA}{{\hat{A}}} \nc{\hB}{{\hat{B}}} \nc{\hC}{{\hat{C}}}
	\nc{\hD}{{\hat{D}}} \nc{\hE}{{\hat{E}}} \nc{\hF}{{\hat{F}}}
	\nc{\hG}{{\hat{G}}} \nc{\hH}{{\hat{H}}} \nc{\hI}{{\hat{I}}}
	\nc{\hJ}{{\hat{J}}} \nc{\hK}{{\hat{K}}} \nc{\hL}{{\hat{L}}}
	\nc{\hM}{{\hat{M}}} \nc{\hN}{{\hat{N}}} \nc{\hO}{{\hat{O}}}
	\nc{\hP}{{\hat{P}}} \nc{\hR}{{\hat{R}}} \nc{\hS}{{\hat{S}}}
	\nc{\hT}{{\hat{T}}} \nc{\hU}{{\hat{U}}} \nc{\hV}{{\hat{V}}}
	\nc{\hW}{{\hat{W}}} \nc{\hX}{{\hat{X}}} \nc{\hZ}{{\hat{Z}}}
	
	\nc{\hn}{{\hat{n}}}
	
	\def\max{\mathop{\rm max}}

	\newcommand{\bra}[1]{\langle#1|}
	\newcommand{\ket}[1]{|#1\rangle}

	\newcommand{\braket}[2]{\langle#1|#2\rangle}

	\usepackage[
	colorlinks,
	linkcolor = blue,
	citecolor = blue,
	urlcolor = blue]{hyperref}
	\def \qed {\hfill \vrule height7pt width 7pt depth 0pt}
	
	\newcounter{lastnote}

\begin{document}
		
		\title{Application of Ramsey theory to localization of  set of  product states via multicopies}

		\author{ Xing-Chen Guo}
\affiliation{ School of Mathematics,
	South China University of Technology, Guangzhou
	510641,  China}

	\author{Mao-Sheng Li}
\email{li.maosheng.math@gmail.com}
\affiliation{ School of Mathematics,
	South China University of Technology, Guangzhou
	510641,  China}

		\date{\today}

		\begin{abstract}
		   It is well known that any $N$   orthogonal pure
		   states can always be perfectly distinguished under local operation and classical communications (LOCC) if $(N-1)$ copies of
		   the state are available [\href{https://doi.org/10.1103/PhysRevLett.85.4972}{Phys. Rev. Lett. \textbf{85}, 4972 (2000)}]. It is important to reduce the number of quantum state copies that ensures the LOCC distinguishability in terms of resource saving and nonlocality strength characterization.   Denote $f_r(N)$ the least number of copies needed to LOCC distinguish any $N$ orthogonal $r$-partite product states. This work will be devoted to the estimation of the upper bound of $f_r(N)$. In fact, we first relate this problem  with   Ramsey theory, a branch of combinatorics dedicated to studying the conditions under which orders must appear.  Subsequently, we prove $f_2(N)\leq \lceil\frac{N}{6}\rceil+2$, which is better than $f_2(N)\leq \lceil\frac{N}{4}\rceil$ obtained in \href{https://doi.org/10.1140/epjp/s13360-021-02100-9}{[Eur. Phys. J. Plus \textbf{136}, 1172 (2021)]} when $N>24$. We further exhibit that for arbitrary $\epsilon>0$, $f_r(N)\leq\lceil\epsilon N\rceil$ always holds for sufficiently large $N$.
		\end{abstract}

		\maketitle                                                                                     
		\section{Introduction}
        Orthogonal pure states can always be perfectly distinguished by means of global measurement \cite{nils}. However, there always be the situation that global measurement is not allowed. In such setting, we may instead perform local measurements on constituent subsystems and communicate the results mutually through classical communications, that is, local operations and classical communications(LOCC), to distinguish those states. But there is a gap between global measurement and LOCC in the sense that some orthogonal states cannot be perfectly distinguished reliably by LOCC, which means the extraction of information of quantum systems may require global operations, manifesting quantum nonlocality \cite{bennett1999quantum}. It is an important topic in quantum information theory that whether or not a set of orthogonal states is LOCC distinguishable. On one hand, local indistinguishability plays an important role in quantum cryptography\cite{hillery1999quantum,terhal2001hiding,rahaman2015quantum}. On the other hand, for  those locally distinguishable states, replacing global operations by LOCC can reduce the cost of the process. 

In Ref.\cite{bennett1999quantum}, Bennett $et\ al.$ constructed a set of orthogonal product states which cannot be reliably distinguished by any sequence of LOCC. This counter-intuitive result reveals ``quantum nonlocality without entanglement” while quantum nonlocality was usually considered as a phenomenon associated with entangled states \cite{Brunner-et-al-2014}. Following this paper, LOCC distinguishability problem of orthogonal quantum states has been studied by lots of researchers \cite{bennett1999unextendible,Wal00,de2004distinguishability,Ghosh-2001,walgate-2002,HSSH,divin03,chen2004distinguishing,rin04,Ghosh-2004,fan-2005,Nathanson-2005,Duan2007,feng09,Duan-2009,BGK-2011,Bandyo-2011,Yu-Duan-2012,Cosentino-2013,BN-2013,Yang13,zhang14,zhang15,Li15,wang15,Yang15,B-IQC-2015,xu16-1,Xu-16-2,zhang16,zhang16-1,Xu-17,Wang-2017-Qinfoprocess,Zhang-Oh-2017,zhang17-1,halder,Li18,Halder19,Zhang1906,Shi20S,Tian20,Xu20b,Halder1909,Halder20c,Li20,yuanj20,Xu20a,Ha21,Zuo21,Yang21,Wang21,Shi21,Yuanj22}. Most of these works are under the assumption that only one copy of states are provided in the LOCC distinguishing protocol.
However, distributing multiplied number of photons in the same time makes it possible to get 
multiplied copies of the states. And we can consume these copies  to distinguish the original orthogonal  states via LOCC. Therefore, LOCC distinguishability of orthogonal states under the condition that multiplied copies are given should also be considered. In this regard, Bandyopadhyay proved that a set of orthogonal mixed states may be LOCC indistinguishable no matter how many copies are supplied \cite{bandyopadhyay2011more}. But for pure states, Walgate \emph{et  al.} showed that it is sufficient to distinguish any $N$ orthogonal pure states by means of LOCC if $(N-1)$ copies are allowed \cite{Wal00}, implying that there might be a relationship between the number of pure states and the localization of those states. As pointed out in Ref. \cite{bennett1999quantum}: ``are
there any sets of states, entangled or not, for which some
finite number (greater than 2) of copies of the state is
necessary for distinguishing the states reliably".  Recently, Banik \emph{et al.} \cite{Banik21}   provided examples of orthonormal bases in two-qubit systems
whose adaptive discrimination (weaker than LOCC) require three copies of the state. On the other hand,   for product states,  Shu \cite{shu2021locality} obtained an upper bound of the least number of copies needed to ensure LOCC distinguishability of any N orthogonal $r$-partite($r\geq2$) product states, which is smaller than $N-1$.  Therefore, it is  interesting to find whether we can reduce the number of quantum state copies that ensure the LOCC distinguishability.    

Following Shu’s paper \cite{shu2021locality}, we consider the LOCC distinguishability of any $N$ orthogonal $r$-partite($r\geq2$) product states under the condition that multiplied copies are given. Denote $f_r(N)$ as the least number of copies needed to ensure the LOCC distinguishability of any $N$ orthogonal $r$-partite product states. Then Shu's result can be reformulated as $f_r(N)\leq\lceil\frac{N}{4}\rceil+1$(and  $f_2(N)\leq\lceil\frac{N}{4}\rceil$) \cite{shu2021locality}, where $\lceil a\rceil$ stands for the least positive integer that not less than $a$. In this paper, we dedicate to reduce this upper bound of $f_r(N)$. We first claim that one of Shu’s lemma is just an application of Ramsey theory, which is a branch of combinatorics. By clarifying the relationship between LOCC distinguishability of set of orthogonal product states and Ramsey numbers, we give a sufficient condition for any $N$ orthogonal product states to exclude $k$ states via LOCC protocols on a single copy, that is, $N$ and $k$ satisfy a series of inequalities about Ramsey numbers. Then we obtain a new upper bound $\lceil\frac{N}{6}\rceil+2$ of $f_2(N)$, which is smaller than Shu’s upper bound $\lceil\frac{N}{4}\rceil$ when the number $N>24$. Generally for multipartite case, we proved that for arbitrary $\epsilon>0$ we have $f_r(N)\leq\lceil\epsilon N\rceil$ when $N$ is sufficiently large. Therefore, our results could be seen as an improvement of Shu's when the cardinality  $N$ of the  set to be distinguished is large enough.  One clue of our paper is, for any orthogonal set of product states, no matter what its structure
is, there must be some substructures that encode the information for localization of the whole set, arising from it once the set is sufficiently large. And Ramsey theory naturally provides a mathematical formalization of this insight.

The rest of this article is organized as follows. In Sec. \ref{sec:Pre}, we review  some related results of Ramsey theory. In Sec. \ref{sec:Bipatite}, we will study the upper bound of $f_2(N)$.  In Sec. \ref{sec:Multipatite}, we will study the upper bound of $f_r(N)$.            Finally, we draw a conclusion  in     section \ref{sec:Conclusion}.

		\section{Preliminaries}\label{sec:Pre}
		In this section, we review some results of Ramsey theory, with an emphasize on Ramsey number, which is crucial in this paper. Most of the proofs of the results listed below can be found in the references  \cite{robertson2021fundamentals,erdos1947some,min1993new,conlon2021lower,sarkozy2011monochromatic,radziszowski2011small}.
		 
		First, we introduce some basic concepts which is  used widely in Ramsey theory. A graph is called a complete graph if any two vertexes of it are connected by a line, that is, an edge of that graph. And we denote the complete graph that has $N$ vertexes as $K_N$. Let $c_1,\cdots,c_r$ be $r$ different colors, if for any edge $l$ of $K_N$, $\exists\ i\in\{1,\cdots,r\}$,  s.t.,  $l$ is colored by $c_i$, then we say $K_N$ is edge $r$-coloring. Then one of the most important result of Ramsey theory is the following Ramsey's Theorem.
		
		\vskip10pt
		\begin{namedtheorem}[Ramsey] \emph{(See Theorem 3.6 of  Ref. \cite{robertson2021fundamentals})}
		    For any positive integers $i_1,\cdots,i_r(r\geq2)$, there must be an $r$-color Ramsey number $R(i_1,\cdots,i_r;r)$ which is the smallest positive integer satisfying the following condition: for any positive integer $N\geq R(i_1,\cdots,i_r;r)$ and any edge $r$-coloring $K_N$, there exists $s\in\{1,\cdots,r\}$, $s.t.$ $K_N$ contains a monochromatic subgraph $K_{i_s}$ whose edges are all colored by the $s$-th color $c_s$. And we use abbreviation $R(i_1,\cdots,i_r)$ for $R(i_1,\cdots,i_r;r)$ and $R_r(m)$ for $R(m,\cdots,m;r)$.
		\end{namedtheorem}

	\begin{proposition}[Robertson\ \cite{robertson2021fundamentals}]\label{Prop:recursive}
	    For any positive integers $k$ and $l$, we have $R(k,l)\leq R(k-1,l)+R(k,l-1)$.
	\end{proposition}
   
	\begin{proposition}[Erd{\"o}s\ \cite{erdos1947some}]\label{Prop:lowerbound}
	    For any positive integer $m$, we have $R(m,m)>2^{\frac{m}{2}}$.
	\end{proposition}
 
	\begin{proposition}[Min\ \cite{min1993new}]\label{Proposition4} 
	    For any positive integer $m$, we have $R_{r+s}(m)>(R_r(m)-1)(R_s(m)-1)$, specifically, for any positive integers $r$ and $m\geq2$ we have $R_{r+1}(m)\geq R_r(m)$.
	\end{proposition}
 
	\begin{proposition}[Conlon {\&} Ferber\ \cite{conlon2021lower}]\label{Proposition5}
	    For any positive integer $m$, we have $R(m,m,m)>3^{\frac{m}{2}}$.
	\end{proposition}
 	\begin{proposition}[S{\'a}rk{\"o}zy\ \cite{sarkozy2011monochromatic}]\label{Proposition6}
	    For any positive integers $m$ and $r$, we have $R(m,3,\cdots,3;r+1)\leq r!m^{r+1}$.
	\end{proposition}

		\section{Cases for bipartite system  }\label{sec:Bipatite}
		    Denote an unknown product state of a given bipartite system by $\ket{\psi}$, and assume that $\ket{\psi}\in\{\ket{a_i}\otimes\ket{b_i}\}_{i=1}^{n}$. Let $A=\{\ket{a_i}\}_{i=1}^{n}$ and $B=\{\ket{b_i}\}_{i=1}^{n}$. It's known that if there are $j$ states in $A$(or $B$) which are orthogonal to each other, without loss of generality, assume that $\{\ket{a_i}\}_{i=1}^{j}$ is an orthogonal set, then a local measurement $\{\ket{a_i}\bra{a_i}\}_{i=1}^{j}\cup\{ {\mathbb{I}-\sum_{i=1}^{j}\ket{a_i}\bra{a_i}}\}$ on a single copy of the states is enough to exclude at least $j-1$ states. Shu showed that $j\geq3$ when $n=6$ through a natural correspondence between 6 orthogonal bipartite product states and edge 2-coloring complete graph $K_6$, which leads to a conclusion that a single copy is enough to exclude 4 bipartite product states when distinguishing 7 orthogonal bipartite product states via LOCC protocols, thus obtaining an upper bound of the number of copies $\lceil\frac{N}{4}\rceil$ that ensures the LOCC distinguishability of any $N$ orthogonal bipartite product states\ \cite{shu2021locality}. In this paper,  we reveal that LOCC distinguishability of orthogonal bipartite product states has a close connection with 2-color Ramsey number, and by which we improve Shu's results\ \cite{shu2021locality}. 
		    
		    First, we review Shu's method of analyzing orthogonal product states by polygons. Regarding the $N$ orthogonal bipartite product states $\{\ket{a_i}\otimes\ket{b_i}\}_{i=1}^{N}$ as the $N$ vertexes of a complete graph $K_N$, then we color the edges of the graph with red and blue as follows.  If $\ket{a_j}\otimes\ket{b_j}$ is orthogonal to $\ket{a_k}\otimes\ket{b_k}$ with respect to $A$, i.e., $\braket{a_j}{a_k}=0$, then color the edge that connects the corresponding two vertexes by red, and color it by blue if $\braket{b_j}{b_k}=0$. Since $\{\ket{a_i}\otimes\ket{b_i}\}_{i=1}^{N}$ is an orthogonal set, each edge of  the graph $K_N$ must either be colored by red or by blue. That is, there is  a correspondence between orthogonal set $\{\ket{a_i}\otimes\ket{b_i}\}_{i=1}^{N}$ and edge 2-coloring $K_N$\ \cite{shu2021locality}. 
		    
		    Shu obtained one of  his  results by proving the lemma that $\{\ket{a_i}\otimes\ket{b_i}\}_{i=1}^{6}$ must have either 3 states in $A$ or 3 states in $B$ which are orthogonal to each other, i.e., edge 2-coloring $K_6$ must contain either a red $K_3$ or a blue $K_3$\ \cite{shu2021locality}. By   Ramsey theory,  this is just equivalent to the fact  $R(3,3)=6$. More generally, an edge 2-coloring $K_{R(m,m)}$ must contain either a red $K_m$ or a blue $K_m$, i.e., $\{\ket{a_i}\otimes\ket{b_i}\}_{i=1}^{R(m,m)}$ must have either $m$ states in $A$ or $m$ states in $B$ that are orthogonal to each other. Therefore, it is sufficient to exclude at least $m-1$ states   via a sequence of local measurements on a single copy. This gives an  improvement of  Shu's Lemma\ \cite{shu2021locality}.  The observation above leads us to the following stronger result. 
		    
		    \begin{theorem}\label{Theorem1}
		       Given  any $  R(m,m)$ orthogonal bipartite product states and  a positive integer $k\geq m$, if for any $0\leq t\leq k-m\ (t\in\mathbb{N})$ the following inequalities  \begin{equation}\label{eq:ineqcondition}
         R(m,m)-(m+t)\geq R(m+1+t,k-m+1-t)\end{equation} are satisfied,
		       then   a single copy   is sufficient to exclude at least $k$ states via  LOCC protocols.
		    \end{theorem}
		    \begin{proof}
		       Denote the $R(m,m)$ orthogonal bipartite product states as $\{\ket{a_i}\otimes\ket{b_i}\}_{i=1}^{R(m,m)}\subseteq\mathbb{C}^{d_A}\otimes\mathbb{C}^{d_B}$, and assume that the unknown state from this set is $\ket{\psi}\in\{\ket{a_i}\otimes\ket{b_i}\}_{i=1}^{R(m,m)}$. Suppose that there are exactly at most $M_A$ and $M_B$ states in $A$ and $B$ which are orthogonal to each other respectively. Define $M=\max\{M_A,M_B\}.$  By the definition of $R(m,m)$, the set $\{\ket{a_i}\otimes\ket{b_i}\}_{i=1}^{R(m,m)}$ must have either $m$ states in $A$ or in $B$ which are orthogonal to each other. Therefore, we have $M\geq m$. Setting  $t=k-m$ in inequality \eqref{eq:ineqcondition}, we have $R(m,m)-k\geq R(k+1,1)=1.$ That is, $R(m,m)\geq k+1$. Therefore, the cardinality of the set  $\{\ket{a_i}\otimes\ket{b_i}\}_{i=1}^{R(m,m)}$ is greater than $k+1$.  Now we separate our proof into two cases according to the number $M$.  
		       \begin{enumerate}[(1)]
		           
		          \item $M\geq k.$ That is, there are at least $k$ states in $A$ or in $B$ which are orthogonal to each other. Without loss of generality, let $\{\ket{a_i}\}_{i=1}^{k}$ be an orthogonal set. Then taking the local measurement  $\{\ket{a_i}\bra{a_i} \}_{i=1}^{k}\cup\{ {\mathbb{I}-\sum_{i=1}^{k}\ket{a_i}\bra{a_i}} \}$ on the first subsystem, the outcome of the measurement have the following different possible cases:
		            \begin{enumerate}[$\bullet$]
		               \item The outcome is not one of $1,\cdots,k$ (that is, the outcome corresponds to $\mathbb{I}-\sum_{i=1}^{k}\ket{a_i}\bra{a_i}$). So we can exclude $k$ states.
		               
		               \item The outcome is one of $1, \cdots, k$  (that is, the outcome corresponds to  one of $\{\ket{a_i}\bra{a_i} \}_{i=1}^{k}$). Without loss of generality, we  assume that the outcome is $k$. Then we can exclude the following $(k-1)$ states $$\ket{a_1}\otimes\ket{b_1}, \cdots, \ket{a_{k-1}}\otimes\ket{b_{k-1}}.$$ What's more, we consider the following two cases:
		               \begin{enumerate}[$\bullet$]
		                   \item If\ $\exists\ j> k$, $s.t.$ $\braket{a_k}{a_j}=0$, then we can also exclude $\ket{a_j}\otimes\ket{b_j}$ via the outcome $k$.
		                   
		                   \item If\ $\forall\ j> k$,\ $\braket{a_k}{a_j}\neq0$, then we must have $\braket{b_k}{b_j}=0$. Especially, we have $\braket{b_k}{b_{k+1}}=0$. Then the second system could take the  local measurement $\{ \ket{b_k}\bra{b_k},  \ket{b_{k+1}}\bra{b_{k+1}} \} \cup\{ {\mathbb{I}-\sum_{i=k}^{k+1}\ket{b_i}\bra{b_i}}\}.$  Each outcome of this measurement  could  exclude  one more state.  Hence, we  could  exclude at least $k$ states totally. 
		              \end{enumerate}
		         \end{enumerate}
		    
		     \item $M= k+1-n ( n\in\mathbb{N} ).$  Notice that $M=k+1-n\geq m$. Therefore, we only need to consider the cases:  $0\leq n\leq k-m+1$. That is, there  are  only $k+1-n$($n\in\mathbb{N}$) states in $A$ or in $B$ which are orthogonal to each other. Without loss of generality, let $\{\ket{a_i}\}_{i=1}^{k+1-n}$ be an orthogonal set.  Now we will prove our statement by induction on $n$. Notice that  the cases $n=0$ and $n=1$       have been   proved  by the previous one case. Now suppose that the result is true for $n\leq N$(where $N+1\leq k-m+1$), i.e., when there are at least $k+1-N$ states in $A$ or in $B$ which are orthogonal to each other, a single copy   is sufficient to exclude at least $k$ states via  LOCC protocols (inductive assumption). 
       
       In the following, we will show that our statement is true   for $n=N+1$. That is,  when there are only  $k+1-(N+1)=k-N$ states in $A$ or in $B$ which are orthogonal to each other, a single copy   is sufficient to exclude at least $k$ states via  LOCC protocols. Without loss of generality, let $\{\ket{a_i}\}_{i=1}^{k-N}$ be an orthogonal set. Taking the local measurement $\{\ket{a_i}\bra{a_i} \}_{i=1}^{k-N}\cup\{ {\mathbb{I}-\sum_{i=1}^{k-N}\ket{a_i}\bra{a_i}} \}$, we have two different possible outcomes: 
		       \begin{enumerate}[(i)]
		           \item The outcome is not among $\{1, \cdots, k-N\}$. Then we exclude $(k-N)$ states $\ket{a_1}\otimes\ket{b_1}, \cdots, \ket{a_{k-N}}\otimes\ket{b_{k-N}}$. The condition 
		           \[R(m,m)-(k-N)\geq R(k-N+1,N+1),\]
		           implies that the number of the remaining states is not less than $R(k-N+1,N+1)$. Regarding the remaining $R(m,m)-(k-N)$ states $\{\ket{a_i}\otimes\ket{b_i}\}_{i=k-N+1}^{R(m,m)}$ as an edge 2-coloring $K_{R(m,m)-(k-N)}$,  by  the definition of Ramsey number, we know that the edge 2-coloring $K_{R(m,m)-(k-N)}$ must contain either a red $K_{k-N+1}$ or a blue $K_{N+1}$.  Therefore,  there must be either $k-N+1$ states in $\{\ket{a_i}\}_{i=k-N+1}^{R(m,m)}$ or $N+1$ states in $\{\ket{b_i}\}_{i=k-N+1}^{R(m,m)}$ which are orthogonal to each other. For the first setting, by the inductive assumption,  we know that our statement is true. For the latter setting, without loss of generality, let $\{\ket{b_i}\}_{i=k-N+1}^{k+1}$ be an orthogonal set. Taking the local measurement for the second subsystem $\{ \ket{b_i}\bra{b_i}\}_{i=k-N+1}^{k+1}\cup\{  {\mathbb{I}-\sum_{i=k-N+1}^{k+1}\ket{b_i}\bra{b_i}}\}, $ we can   exclude at least $N$ more states for each possible outcome.  Hence, we could exclude at least $k-N+N=k$ states in total.
		           \item The outcome is one of $1, \cdots, k-N$. Without loss of generality, we  assume that the outcome is $k-N$. So we can exclude the $k-N-1$ states $\ket{a_1}\otimes\ket{b_1}, \cdots, \ket{a_{k-N-1}}\otimes\ket{b_{k-N-1}}$. Furthermore, we consider the following different cases:
		           \begin{enumerate}[$\bullet$]
		                
                  \item If there are only $s$ (where $1\leq s\leq N+1$) states in $\{\ket{a_i}\}_{i=k-N+1}^{R(m,m)}$ which are orthogonal to $\ket{a_{k-N}}$, without loss of generality, let \[\braket{a_{k-N}}{a_{k-N+1}}=\cdots=\braket{a_{k-N}}{a_{k-N+s}}=0.\]
                  Then for the outcome $k-N$, we can  also exclude $s$ more states $$\ket{a_{k-N+1}}\otimes\ket{b_{k-N+1}}, \cdots, \ket{a_{k-N+s}}\otimes\ket{b_{k-N+s}}.$$ 
                  Moreover,   setting $t=k-m-(N+1-s)$ in inequality \eqref{eq:ineqcondition}, we have 
		               \[R(m,m)-(k-N-1+s)\geq R(k-N+s, N+2-s).\]
		              This  implies that the number of the remaining states $\{\ket{a_{k-N}}\otimes\ket{b_{k-N}},\ket{a_{i}}\otimes\ket{b_{i}}\}_{i=k-N+s+1}^{R(m,m)}$ is not less than $R(k-N+s, N+2-s)$. By the definition of Ramsey number,   there are either $k-N+s$ states in $\{\ket{a_{k-N}},\ket{a_{i}}\}_{i=k-N+s+1}^{R(m,m)}$ or $N+2-s$ states in $\{\ket{b_{k-N}},\ket{b_{i}}\}_{i=k-N+s+1}^{R(m,m)}$ that orthogonal to each other.  For the first setting, as $s\geq 1$, by the inductive assumption,  we know that our statement is true. For the latter setting, without loss of generality, let $\{|b_{k-N}\rangle\}\cup \{\ket{b_i}\}_{i=k-N+s+1}^{k+1}$ be an orthogonal set. Taking the local measurement  $$\{ \ket{b_i}\bra{b_i}\}_{i=k-N+s+1}^{k+1}\cup\{  { \ket{b_{k-N}}\bra{b_{k-N}},\mathbb{I}- \ket{b_{k-N}}\bra{b_{k-N}}-\sum_{i=k-N+s+1}^{k+1}\ket{b_i}\bra{b_i}}\} $$ 
                for the second subsystem,   we can   exclude at least $(N+1-s)$ more states for each possible outcome.  Hence, we could exclude at least $(k-N-1)+s+(N+1-s)=k$ states in total.

		              \item If there is no state in $\{\ket{a_i}\}_{i=k-N+1}^{R(m,m)}$ that is orthogonal to $\ket{a_{k-N}}$. Then for $\forall\ k-N+1\leq i\leq R(m,m)$ we have $\braket{b_{k-N}}{b_i}=0$, and from the condition
		              \[R(m,m)-(k-N)\geq R(k-N+1,N+1),\]
		              there must be $k-N+1$ states in $\{\ket{a_i}\}_{i=k-N+1}^{R(m,m)}$ or $N+1$ states in $\{\ket{b_i}\}_{i=k-N+1}^{R(m,m)}$ which are orthogonal to each other. For the first setting, by the inductive assumption,  we know that our statement is true. For the latter setting, without loss of generality, let $\{\ket{b_i}\}_{i=k-N+1}^{k+1}$ be an orthogonal set, that is, $\{\ket{b_i}\}_{i=k-N}^{k+1}$ is an orthogonal set. Then taking the local measurement $\{ \ket{b_i}\bra{b_i}\}_{i=k-N}^{k+1}\cup\{  {\mathbb{I}-\sum_{i=k-N}^{k+1}\ket{b_i}\bra{b_i}}\}$  for the second subsystem, it is sufficient to exclude at least $(N+1)$ more states for each possible outcome. Hence, we can  exclude $(k-N-1)+(N+1)=k$ states totally.

		       \end{enumerate}
		     \end{enumerate}
		     By induction we can conclude that the result is true for all cases.

         \end{enumerate}
		    \end{proof} 
		    \vskip 15pt
		    Note that for any positive integer $N\geq R(m,m)$ and $k\geq m$, if $m$ and $k$ satisfy the condition of \autoref{Theorem1}, then a single copy is sufficient to exclude $k$ states when distinguishing any $N$ orthogonal bipartite product states via LOCC protocols. We can see that this statement is true, just by considering $R(m,m)$ states among $N$ states. Similarly in the rest of the paper, any statement like ``for any $x$ states $\cdots$ to exclude $y$ states $\cdots$" just implies  that ``for any $N$($N\geq x$) states $\cdots$ to exclude $y$ states $\cdots$", such as \autoref{Corollary1}, \autoref{Theorem3} and \autoref{Corollary2}.
      
            And from \autoref{Theorem1}, it's easy to prove the following corollary using the known results of 2-color Ramsey number. 
		    \vskip 15pt
		    \begin{corollary}\label{Corollary1}
		       (1) For any $R(4,4)=18$ orthogonal bipartite product states, a single copy   is sufficient to exclude  $6$ states via  LOCC protocols; (2) For any $R(m,m)$ orthogonal bipartite product states$(m\geq5, m\in\mathbb{N^{+}})$, a single copy   is sufficient to exclude  $m+3$ states via  LOCC protocols.
		    \end{corollary}
		    \begin{proof}
		       Just need to check the relevant inequalities of \autoref{Theorem1}.
		       \begin{enumerate}[(1)]
		           \item Set $m=4, k=6.$ We only need to check directly
		              \begin{align*}
		                  R(4,4)-4 &\geq R(5,3),\\
		                  R(4,4)-5 &\geq R(6,2),\\
		                  R(4,4)-6 &\geq R(7,1)
		              \end{align*}
		              from the known values of the 2-color Ramsey number(see \autoref{table1}).
		           \item Set $k=m+3$. We only need to check
		              \begin{equation*}
                      \begin{array}{l}
                         R(m,m)-m\geq R(m+1,4),\\
                         R(m,m)-(m+1)\geq R(m+2,3),\\
                         R(m,m)-(m+2)\geq R(m+3,2),\\
                         R(m,m)-(m+3)\geq R(m+4,1).
                      \end{array}
                      \end{equation*}
                      \begin{enumerate} 
                          \item  The last two inequities. By \autoref{Prop:lowerbound}, we have  $R(m,m)>2^{\frac{m}{2}}$.  It's easy to check that $$2^{\frac{m}{2}}>2m+7$$ when $m\geq10$. As $R(m+3,2)=m+3,R(m+4,1)=m+4,$ so  the last two inequities are  satisfied when $m\geq 10$. For the other cases, the two inequalities are also satisfied by checking  the values in  \autoref{table1}.
                          \item Now we consider the inequality $R(m,m)-(m+1)\geq R(m+2,3)$.  By \autoref{Prop:recursive}, we know  that $R(m,n)\leq R(m-1,n)+R(m,n-1)$.  So we have
                             \begin{align*}
                                R(m+2,3)&\leq R(m+1,3)+R(m+2,2)\\
                                        &=R(m+1,3)+m+2\\
                                        &\leq R(m,3)+m+1+m+2\\
                                        &\leq \cdots\\
                                        &\leq R(5,3)+\frac{(6+m+2)(m-3)}{2}\\
                                        &=\frac{1}{2}m^2+\frac{5}{2}m+2 \text{.}
                            \end{align*}\\
                          Therefore, $R(m+2,3)+m+1\leq\frac{1}{2}m^2+\frac{7}{2}m+3$. Notice that $R(m,m)>2^{\frac{m}{2}}$, it's easy to show that when $m\geq15$ we have
                          \[R(m,m)>2^{\frac{m}{2}}>\frac{1}{2}m^2+\frac{7}{2}m+3\geq R(m+2,3)+m+1.\]
                          For $5\leq m\leq 14$, we can  show that the inequality $R(m,m)-(m+1)\geq R(m+2,3)$   also holds by comparing the known lower bound of $R(m,m)$ with the known upper bound of the corresponding $R(m+2,3)$, see \autoref{table1}, \autoref{table2} and \autoref{table3}.
                        \item Now we consider the 
 first inequality $R(m,m)-m\geq R(m+1,4)$. Similarly, from the fact that $R(m,n)\leq R(m-1,n)+R(m,n-1)$,  we  obtain that
                           \begin{align*}
                              R(m+1,4)&\leq R(m+1,3)+R(m,4)\\
                                      &\leq\frac{(6+m+1)(m-4)}{2}+R(5,3)+R(m,4)\\
                                      &\leq\frac{1}{2}m^2+\frac{3}{2}m+R(m,4).
                           \end{align*}
                        So we have 
                           \begin{align*}
                              R(m+1,4) &\leq\frac{1}{2}m^2+\frac{3}{2}m+R(m,4)\\
                                       &\leq\cdots\\
                                       &\leq \frac{1}{2}(m^2+\cdots+5^2)+\frac{3}{2}(m+\cdots+5)+R(5,4)\\
                                       &=\frac{1}{6}m^3+m^2+\frac{5}{6}m-5 \text{.}
                           \end{align*}\\
                        Hence, we have $R(m+1,4)+m\leq\frac{1}{6}m^3+m^2+\frac{11}{6}m-5$.    It can be shown that  the inequality \begin{equation*}
                 R(m,m)>2^{\frac{m}{2}}>\frac{1}{6}m^3+m^2+\frac{11}{6}m-5\geq R(m+1,4)+m \end{equation*} holds when $m\geq23$.
                        For $5\leq m\leq22$, the known lower bounds of $R(m,m)$ which are illustrated in \autoref{table2} show that: when $14\leq m\leq 22$, $R(m,m)\geq R(14,14)>\frac{1}{6}\cdot23^3+23^2+\frac{11}{6}\cdot23-5$. It  implies that the inequity $R(m,m)-m\geq R(m+1,4)$ holds. For $5\leq m\leq 13$, we just need to compare the known lower bound of $R(m,m)$ with the upper bound $\frac{1}{6}m^3+m^2+\frac{11}{6}m-5$ of the corresponding $R(m+1,4)+m$(see \autoref{table1} and \autoref{table2}).
                        \end{enumerate}
                \end{enumerate}
		    \end{proof}
		    
\begin{table}[H]
\centering
\scalebox{1}{
\begin{tabular}{|c|c|c|c|c|c|c|c|c|c|c|c|c|c|}
\hline
\diagbox{k}{l}   & 3          & 4           & 5                                               & 6                                                 & 7                                                 & 8                                                  & 9                                                  & 10                                                                     & 11                                                  & 12                                                 & 13                                                  & 14                                                 & 15                                                  \\ \hline
3  & \textbf{6} & \textbf{9}  & \textbf{14}                                     & \textbf{18}                                       & \textbf{23}                                       & \textbf{28}                                        & \textbf{36}                                        & {
	 \begin{tabular}[c]{@{}c@{}}40\\ 42\end{tabular}} & \begin{tabular}[c]{@{}c@{}}47\\ 50\end{tabular}     & \begin{tabular}[c]{@{}c@{}}53\\ 59\end{tabular}    & \begin{tabular}[c]{@{}c@{}}60\\ 68\end{tabular}     & \begin{tabular}[c]{@{}c@{}}67\\ 77\end{tabular}    & \begin{tabular}[c]{@{}c@{}}74\\ 87\end{tabular}     \\ \hline
4  &            & \textbf{18} & \textbf{25}                                     & \begin{tabular}[c]{@{}c@{}}36\\ 41\end{tabular}   & \begin{tabular}[c]{@{}c@{}}49\\ 61\end{tabular}   & \begin{tabular}[c]{@{}c@{}}59\\ 84\end{tabular}    & \begin{tabular}[c]{@{}c@{}}73\\ 115\end{tabular}   & \begin{tabular}[c]{@{}c@{}}92\\ 149\end{tabular}                       & \begin{tabular}[c]{@{}c@{}}102\\ 191\end{tabular}   & \begin{tabular}[c]{@{}c@{}}128\\ 238\end{tabular}  & \begin{tabular}[c]{@{}c@{}}138\\ 291\end{tabular}   & \begin{tabular}[c]{@{}c@{}}147\\ 349\end{tabular}  & \begin{tabular}[c]{@{}c@{}}158\\ 417\end{tabular}   \\ \hline
5  &            &             & \begin{tabular}[c]{@{}c@{}}43\\ 48\end{tabular} & \begin{tabular}[c]{@{}c@{}}58\\ 87\end{tabular}   & \begin{tabular}[c]{@{}c@{}}80\\ 143\end{tabular}  & \begin{tabular}[c]{@{}c@{}}101\\ 216\end{tabular}  & \begin{tabular}[c]{@{}c@{}}133\\ 316\end{tabular}  & \begin{tabular}[c]{@{}c@{}}149\\ 442\end{tabular}                      & \begin{tabular}[c]{@{}c@{}}183\\ 633\end{tabular}   & \begin{tabular}[c]{@{}c@{}}203\\ 848\end{tabular}  & \begin{tabular}[c]{@{}c@{}}233\\ 1138\end{tabular}  & \begin{tabular}[c]{@{}c@{}}267\\ 1461\end{tabular} & \begin{tabular}[c]{@{}c@{}}275\\ 1878\end{tabular}  \\ \hline
6  &            &             &                                                 & \begin{tabular}[c]{@{}c@{}}102\\ 165\end{tabular} & \begin{tabular}[c]{@{}c@{}}115\\ 298\end{tabular} & \begin{tabular}[c]{@{}c@{}}134\\ 495\end{tabular}  & \begin{tabular}[c]{@{}c@{}}183\\ 780\end{tabular}  & \begin{tabular}[c]{@{}c@{}}204\\ 1171\end{tabular}                     & \begin{tabular}[c]{@{}c@{}}262\\ 1804\end{tabular}  & \begin{tabular}[c]{@{}c@{}}294\\ 2566\end{tabular} & \begin{tabular}[c]{@{}c@{}}347\\ 3703\end{tabular}  & \begin{tabular}[c]{@{}c@{}}L\\ 5033\end{tabular}   & \begin{tabular}[c]{@{}c@{}}401\\ 6911\end{tabular}  \\ \hline
7  &            &             &                                                 &                                                   & \begin{tabular}[c]{@{}c@{}}205\\ 540\end{tabular} & \begin{tabular}[c]{@{}c@{}}219\\ 1031\end{tabular} & \begin{tabular}[c]{@{}c@{}}252\\ 1713\end{tabular} & \begin{tabular}[c]{@{}c@{}}292\\ 2826\end{tabular}                     & \begin{tabular}[c]{@{}c@{}}405\\ 4553\end{tabular}  & \begin{tabular}[c]{@{}c@{}}417\\ 6954\end{tabular} & \begin{tabular}[c]{@{}c@{}}511\\ 10578\end{tabular} & \begin{tabular}[c]{@{}c@{}}L\\ 15263\end{tabular}  & \begin{tabular}[c]{@{}c@{}}L\\ 22112\end{tabular}   \\ \hline
8  &            &             &                                                 &                                                   &                                                   & \begin{tabular}[c]{@{}c@{}}282\\ 1870\end{tabular} & \begin{tabular}[c]{@{}c@{}}329\\ 3583\end{tabular} & \begin{tabular}[c]{@{}c@{}}343\\ 6090\end{tabular}                     & \begin{tabular}[c]{@{}c@{}}457\\ 10630\end{tabular} & \begin{tabular}[c]{@{}c@{}}L\\ 16944\end{tabular}  & \begin{tabular}[c]{@{}c@{}}817\\ 27485\end{tabular} & \begin{tabular}[c]{@{}c@{}}L\\ 41525\end{tabular}  & \begin{tabular}[c]{@{}c@{}}873\\ 63609\end{tabular} \\ \hline
9  &            &             &                                                 &                                                   &                                                   &                                                    & \begin{tabular}[c]{@{}c@{}}565\\ 6588\end{tabular} & \begin{tabular}[c]{@{}c@{}}581\\ 12677\end{tabular}                    & \begin{tabular}[c]{@{}c@{}}L\\ 22325\end{tabular}   & \begin{tabular}[c]{@{}c@{}}L\\ 38832\end{tabular}  & \begin{tabular}[c]{@{}c@{}}L\\ 64864\end{tabular}   & \begin{tabular}[c]{@{}c@{}}\\ \end{tabular}      & \begin{tabular}[c]{@{}c@{}}\\ \end{tabular}       \\ \hline
10 &            &             &                                                 &                                                   &                                                   &                                                    &                                                    & \begin{tabular}[c]{@{}c@{}}798\\ 23556\end{tabular}                    & \begin{tabular}[c]{@{}c@{}}L\\ 45881\end{tabular}   & \begin{tabular}[c]{@{}c@{}}L\\ 81123\end{tabular}  & \begin{tabular}[c]{@{}c@{}}\\ \end{tabular}       & \begin{tabular}[c]{@{}c@{}}\\ \end{tabular}      & \begin{tabular}[c]{@{}c@{}}1313\\ U\end{tabular}    \\ \hline
\end{tabular}
}
\caption{Known values(bold), lower bounds(above) and upper bounds(below) of $R(k,l)$ for $k\leq10$,\ $l\leq15$ \cite{radziszowski2011small}.}
\label{table1}
\end{table}

\vskip10pt

\begin{table}[H]
\centering
\scalebox{1}{
\begin{tabular}{|c|c|c|c|c|c|c|c|c|c|c|c|c|c|}
\hline
k           & 11   & 12   & 13   & 14   & 15   & 16   & 17   & 18    & 19    & 20    & 21    & 22    & 23    \\ \hline
lower bound & 1597 & 1640 & 2557 & 2989 & 5485 & 5605 & 8917 & 11005 & 17885 & 21725 & 30925 & 39109 & 49421 \\ \hline
\end{tabular}
}
\caption{Known lower bounds of $R(k,k)$ for\ $k\geq11$\cite{radziszowski2011small}.}
\label{table2}
\end{table}

\vskip10pt

\begin{table}[H]
\centering
\scalebox{1}{
\begin{tabular}{|c|c|c|c|c|c|c|c|c|c|}
\hline
\diagbox{k}{l}& 15                                              & 16                                              & 17                                               & 18                                               & 19                                                & 20                                                & 21                                                & 22                                                & 23                                                \\ \hline
3 & \begin{tabular}[c]{@{}c@{}}74\\ 87\end{tabular} & \begin{tabular}[c]{@{}c@{}}82\\ 97\end{tabular} & \begin{tabular}[c]{@{}c@{}}92\\ 109\end{tabular} & \begin{tabular}[c]{@{}c@{}}99\\ 120\end{tabular} & \begin{tabular}[c]{@{}c@{}}106\\ 132\end{tabular} & \begin{tabular}[c]{@{}c@{}}111\\ 145\end{tabular} & \begin{tabular}[c]{@{}c@{}}122\\ 157\end{tabular} & \begin{tabular}[c]{@{}c@{}}131\\ 171\end{tabular} & \begin{tabular}[c]{@{}c@{}}139\\ 185\end{tabular} \\ \hline
4 & \textbf{158}                                    & \textbf{170}                                    & \textbf{200}                                     & \textbf{205}                                     & \textbf{213}                                      & \textbf{234}                                      & \textbf{242}                                      & \textbf{314}                                      &                                                   \\ \hline
5 & \textbf{275}                                    & \textbf{293}                                    & \textbf{388}                                     & \textbf{396}                                     & \textbf{411}                                      & \textbf{424}                                      & \textbf{441}                                      & \textbf{492}                                      & \textbf{521}                                      \\ \hline
6 & \textbf{401}                                    & \textbf{434}                                    & \textbf{548}                                     & \textbf{614}                                     & \textbf{710}                                      & \textbf{878}                                      & \textbf{888}                                      & \textbf{1070}                                     &                                                   \\ \hline
7 &                                                 & \textbf{629}                                    & \textbf{729}                                     & \textbf{797}                                     & \textbf{908}                                      &                                                   & \textbf{1214}                                     &                                                   &                                                   \\ \hline
8 & \textbf{873}                                    &                                                 & \textbf{1005}                                    & \textbf{1049}                                    & \textbf{1237}                                     &                                                   & \textbf{1617}                                     &                                                   &                                                   \\ \hline
\end{tabular}
}
\caption{bounds of $R(k,l)$, Lower bounds(above) and upper bounds(below) are given for $k=3$, only lower bounds(bold) for $k\geq4$ \cite{radziszowski2011small}.}
\label{table3}
\end{table}
            
		  From our results,   it is possible to exclude more than 4 states via LOCC protocols with a single copy  of  $N$    orthogonal bipartite product states $\{\ket{a_i}\otimes\ket{b_i}\}_{i=1}^{N}$    when  $N$ is  large enough.  Before going further, we give the following definition.
		  \begin{definition}\label{Definition1}
		     For positive integers $N,\ r\geq2$, we define $f_r(N)$ as the least number of copies  which ensures the LOCC distinguishability of any $N$ orthogonal $r$-partite product states. 
		  \end{definition} 
		 So $f_r(N)$ is the unique positive integer that satisfies the following two conditions:
		     \begin{enumerate}[(i)]
		         \item any $N$ orthogonal $r$-partite product states are LOCC distinguishable provided  $f_r(N)$ copies of given states;
		         \item there exist $N$ orthogonal $r$-partite product states which are   LOCC indistinguishable even provided  $f_r(N)-1$ copies of given states.
		     \end{enumerate}  
		   Shu   \cite{shu2021locality} obtained an upper bound of $f_r(N)$. Especially,   he proved  that $f_2(N)\leq \lceil\frac{N}{4}\rceil$. In the remainder of this section, we dedicate to obtain a smaller upper bound of $f_2(N)$ for large $N$ by applying our results obtained above.
		 \begin{theorem}\label{Theorem3}
		    $f_2(N)\leq\lceil\frac{N}{6}\rceil+2$.
		 \end{theorem}
		 \begin{proof}
		     First, by Shu's result, $f_2(N)\leq \lceil\frac{N}{4}\rceil.$ We can check that  $\lceil\frac{N}{4}\rceil\leq \lceil\frac{N}{6}\rceil+2$ when $N\leq 17.$ For every integer $N\geq 18,$  it can be written as 
       $$ N=6 q+r,\  q,r \in \mathbb{N},  q\geq 3, \text{and   }  0\leq r\leq 5.$$
By \autoref{Corollary1}, if $q\geq 3,$
         we can exclude 6 states via a single copy.  And so we
can use $q-2$ copies to exclude $6(q-2)$ states and left $ 12 + r$ states. To distinguish the remaining $12+r$ states,
$$\lceil\frac{12 + r}{4}\rceil$$
more copies are enough. Therefore,  the total number of copies needed are at most 
$$ (q-2)+\lceil\frac{12 + r}{4}\rceil\leq (q-2)+\lceil\frac{12 + r}{6}\rceil+2= q+\lceil\frac{  r}{6}\rceil+2=\lceil\frac{N}{6}\rceil+2.$$
		 \end{proof}

	 	\section{Cases for multipartite system}\label{sec:Multipatite}
	 	    In this section, we consider the more general $r$-partite system($r\geq2$). We claim that our results of bipartite system can be naturally generalized to $r$-partite case, by contemplating the correspondence between $N$ orthogonal $r$-partite product states and edge $r$-coloring $K_N$, which allows us to make use of the language of $r$-color Ramsey number, similar to the bipartite case.
	 	    
	 	    For convenience, we first clarify some notations used in the rest of the paper. 
	 	    \vskip 10pt
	 	    \paragraph*{Notations}
	 	    \begin{enumerate}[(1)]
	 	        \item “WLG” stands for “Without loss of generality”;
	 	        \item Let $S$ be an $r$-partite system which is composed of $r$ subsystems $S_1, \cdots, S_r$. Assume that the state of $S$ is prepared in one of $N$ orthogonal $r$-partite product states $\{\ket{\psi_i}=\ket{\psi_{i}^{1}}\otimes\cdots\otimes\ket{\psi_{i}^{r}}\}_{i=1}^{N}$, where $\ket{\psi_{i}^{j}}$ corresponds to the $j$-th subsystem $S_j$. We say the $k$ product states $\ket{\psi_{n_1}}, \cdots, \ket{\psi_{n_k}}$ are orthogonal to each other with respect to $S_j$ if $\ket{\psi_{n_1}^j}, \cdots, \ket{\psi_{n_k}^j}$ are orthogonal to each other;
	 	        \item If $\ket{\psi_{n_1}}, \cdots, \ket{\psi_{n_k}}$ are orthogonal to each other with respect to $S_j$, then we use symbol $M_{n_1,\cdots,n_k}^j$ to denote the local measurement \[\{\ket{\psi_{n_i}^j}\bra{\psi_{n_i}^j}\}_{i=1}^{k}\cup\{{\mathbb{I}-\sum_{i=1}^{k}\ket{\psi_{n_i}^j}\bra{\psi_{n_i}^j}}\}\]
	 	        on the $j$-th subsystem $S_j$.
	 	    \end{enumerate}
	 	    
	 	    Now we illuminate the correspondence between $N$ orthogonal $r$-partite product states and edge $r$-coloring $K_N$ in detail, as well as introducing the usage of the language of $r$-color Ramsey number, just as the bipartite case. Regard the $N$ orthogonal $r$-partite product states $\{\ket{\psi_i}=\ket{\psi_{i}^{1}}\otimes\cdots\otimes\ket{\psi_{i}^{r}}\}_{i=1}^{N}$ as the $N$ vertexes of a complete graph $K_N$, and if two states are orthogonal to each other with respect to the $j$-th subsystem $S_j$, then color the edge that connects the corresponding two vertexes by the $j$-th color. So each edge of the corresponding $K_N$ must be colored by one of the $r$ colors, that is, there is a correspondence between $\{\ket{\psi_i}=\ket{\psi_{i}^{1}}\otimes\cdots\otimes\ket{\psi_{i}^{r}}\}_{i=1}^{N}$ and edge $r$-coloring $K_N$. Then $N\geq R(i_1,\cdots,i_r)$ means that there exists $k\in\{1,\cdots,r\}$, s.t., $K_N$ contains a subgraph $K_{i_k}$ whose edges are all colored by the $k$-th color. That is, there exists a subsystem $S_k$, s.t., there exist $i_k$ states in $\{\ket{\psi_i}=\ket{\psi_{i}^{1}}\otimes\cdots\otimes\ket{\psi_{i}^{r}}\}_{i=1}^{N}$ which are orthogonal to each other with respect to $S_k$.
	 	    
	 	    With the statements above, we can naturally derive the $r$-partite versions of \autoref{Theorem1} and \autoref{Corollary1}. The techniques applied to prove the $r$-partite results are just similar to those used in the bipartite case.
	 	    
	 	    \begin{theorem}\label{Theorem3}
	 	        Given any $R_r(m)$ orthogonal $r$-partite product states and a positive integer $k\geq m$, if for any $ 0\leq t\leq k-m\ (t\in\mathbb{N})$ the following inequalities 
	 	        \begin{equation}\label{eq:ineqcondition2}R_r(m)-(m+t)\geq R(m+1+t,k-m+1-t,\cdots,k-m+1-t;\ r),\end{equation} are satisfied, then a single copy is sufficient to exclude at least $k$ states via LOCC protocols. 
	 	    \end{theorem}
	 	    
	 	    \begin{proof}
	 	          Denote the $R_r(m)$ orthogonal $r$-partite product states as $\{\ket{\psi_i}=\ket{\psi_{i}^{1}}\otimes\cdots\otimes\ket{\psi_{i}^{r}}\}_{i=1}^{R_r(m)}$, and assume that the unknown state from this set is $\ket{\psi}\in\{\ket{\psi_i}\}_{i=1}^{R_r(m)}$. Suppose that there are exactly $M_j$ states which are orthogonal to each other with respect to subsystem $S_j$, $j=1,\cdots,r$. Define $M=\max\{M_j\}_{j=1}^{r}$. By the definition of $R_r(m)$, there must be a subsystem $S_k$, s.t., the set $\{\ket{\psi_i}\}_{i=1}^{R_r(m)}$ must have $m$ states which are orthogonal to each other with respect to $S_k$. So we have $M\geq m$. Setting $t=k-m$ in inequality \eqref{eq:ineqcondition2}, we have $R_r(m)-k\geq R(k+1,1,\cdots,1;r)=1$, that is, $R_r(m)\geq k+1$. Therefore, the cardinality of the set $\{\ket{\psi_i}\}_{i=1}^{R_r(m)}$ is greater than $k+1$. Now we separate our proof into two cases according to the number $M$. 
	 	          \begin{enumerate}[(1)]
		           
		          \item $M\geq k.$ That is, there must be a subsystem $s.t.$ there are at least $k$ states which are orthogonal to each other with respect to that subsystem. WLG let $\{\ket{\psi_i^1}\}_{i=1}^{k}$ be an orthogonal set. Then taking $M_{1,\cdots,k}^1$, the outcome of the measurement have the following different possible cases:
		            \begin{enumerate}[$\bullet$]
		               \item The outcome is not one of $1,\cdots,k$ (that is, the outcome corresponds to $\mathbb{I}-\sum_{i=1}^{k}\ket{\psi_{i}^1}\bra{\psi_{i}^1}$). So we can exclude $k$ states.
		               
		               \item The outcome is one of $1, \cdots, k$  (that is, the outcome corresponds to  one of $\{\ket{\psi_i^1}\bra{\psi_i^1} \}_{i=1}^{k}$). WLG, we assume that the outcome is $k$. Then we can exclude the following $(k-1)$ states $$\ket{\psi_1},\cdots,\ket{\psi_{k-1}}.$$ What's more, we consider the following two cases:
		               \begin{enumerate}[$\bullet$]
		                   \item If\ $\exists\ j> k$, s.t., $\braket{\psi_k^1}{\psi_j^1}=0$, then we can also exclude $\ket{\psi_j}$ via the outcome $k$.
		                   
		                   \item If\ $\forall\ j> k$,\ $\braket{\psi_k^1}{\psi_j^1}\neq0$, then there must exist $s\in\{2,\cdots,r\}$, s.t., $\braket{\psi_k^s}{\psi_{k+1}^s}=0$. Then taking $M_{k,k+1}^s,$  each outcome of this measurement  could  exclude  one more state.  Hence, we  could  exclude at least $k$ states totally. 
		              \end{enumerate}
		         \end{enumerate}
		    
		     \item $M= k+1-n ( n\in\mathbb{N} ).$  Notice that $M=k+1-n\geq m$. Therefore, we only need to consider the cases:  $0\leq n\leq k-m+1$. That is, there must be a subsystem, s.t., there are $k+1-n$($n\in\mathbb{N}$) states which are orthogonal to each other with respect to that subsystem. WLG let $\{\ket{\psi_i^1}\}_{i=1}^{k+1-n}$ be an orthogonal set.  Now we will prove our statement by induction on $n$. Notice that  the cases $n=0$ and $n=1$       have been   proved  by the previous one case. Now suppose that the result is true for $n\leq N$(where $N+1\leq k-m+1$), i.e., when $M=\max\{M_j\}_{j=1}^r=k+1-N$, a single copy is sufficient to exclude at least $k$ states via  LOCC protocols (inductive assumption). 
       
       In the following, we will show that our statement is true   for $n=N+1$. That is, when $M=\max\{M_j\}_{j=1}^r=k+1-(N+1)=k-N$, a single copy is sufficient to exclude at least $k$ states via LOCC protocols. WLG let $\{\ket{\psi_i^1}\}_{i=1}^{k-N}$ be an orthogonal set. Taking $M_{1,\cdots,k-N}^1$, we have two different possible outcomes: 
		       \begin{enumerate}[(i)]
		           \item The outcome is not among $\{1, \cdots, k-N\}$. Then we exclude $(k-N)$ states $\ket{\psi_1}, \cdots, \ket{\psi_{k-N}}$. The condition 
		           \[R_r(m)-(k-N)\geq R(k-N+1,N+1,\cdots,N+1;r),\]
		           implies that the number of the remaining states is not less than $R(k-N+1,N+1,\cdots,N+1;r)$. Regarding the remaining $R_r(m)-(k-N)$ states $\{\ket{\psi_i}\}_{i=k-N+1}^{R_r(m)}$ as an edge r-coloring $K_{R_r(m)-(k-N)}$,  by  the definition of Ramsey number, we know that there must be either $k-N+1$ states which are orthogonal to each other with respect to $S_1$, or $\exists\ q\in\{2,\cdots,r\}$, s.t., there are $N+1$ states which are orthogonal to each other with respect to $S_q$. For the first setting, by the inductive assumption,  we know that our statement is true. For the latter setting, WLG let $\{\ket{\psi_i^q}\}_{i=k-N+1}^{k+1}$ be an orthogonal set. Taking $M_{k-N+1,\cdots,k+1}^q,$ we can  exclude at least $N$ more states for each possible outcome.  Hence, we could exclude at least $k-N+N=k$ states in total.
		           \item The outcome is one of $1, \cdots, k-N$. WLG assume that the outcome is $k-N$. So we can exclude the $k-N-1$ states $\ket{\psi_1}, \cdots, \ket{\psi_{k-N-1}}$. Furthermore, we consider the following different cases:
		           \begin{enumerate}[$\bullet$]
		                
                  \item If there are only $s$ (where $1\leq s\leq N+1$) states in $\{\ket{\psi_i^1}\}_{i=k-N+1}^{R_r(m)}$ which are orthogonal to $\ket{\psi_{k-N}^1}$, WLG let \[\braket{\psi_{k-N}^1}{\psi_{k-N+1}^1}=\cdots=\braket{\psi_{k-N}^1}{\psi_{k-N+s}^1}=0.\]
                  Then for the outcome $k-N$, we can  also exclude $s$ more states $$\ket{\psi_{k-N+1}}, \cdots,\ket{\psi_{k-N+s}}.$$ 
                  Moreover,   setting $t=k-m-(N+1-s)$ in inequality \eqref{eq:ineqcondition2}, we have 
		               \[R_r(m)-(k-N-1+s)\geq R(k-N+s, N+2-s,\cdots,N+2-s;r).\]
		              This  implies that the number of the remaining states $\{\ket{\psi_{k-N}},\ket{\psi_{i}}\}_{i=k-N+s+1}^{R_r(m)}$ is not less than $R(k-N+s, N+2-s,\cdots,N+2-s;r)$. By the definition of Ramsey number,   there are either $k-N+s$ states in $\{\ket{\psi_{k-N}},\ket{\psi_{i}}\}_{i=k-N+s+1}^{R_r(m)}$ which are orthogonal to each other with respect to $S_1$, or  $\exists\ q\in\{2,\cdots,r\}$, s.t., there are $N+2-s$ states in $\{\ket{\psi_{k-N}},\ket{\psi_{i}}\}_{i=k-N+s+1}^{R_r(m)}$ that are orthogonal to each other with respect to $S_q$.  For the first setting, as $s\geq 1$, by the inductive assumption,  we know that our statement is true. For the latter setting, WLG let $\{\ket{\psi_{k-N}^q}\}\cup\{\ket{\psi_i^q}\}_{i=k-N+s+1}^{k+1}$ be an orthogonal set. Taking  $M_{k-N,k-N+s+1,k-N+s+2,\cdots,k+1}^q,$   we can   exclude at least $(N+1-s)$ more states for each possible outcome.  Hence, we could exclude at least $(k-N-1)+s+(N+1-s)=k$ states in total.

		              \item If there is no state in $\{\ket{\psi_i^1}\}_{i=k-N+1}^{R_r(m)}$ that is orthogonal to $\ket{\psi_{k-N}^1}$. Then from the condition
		              \[R_r(m)-(k-N)\geq R(k-N+1,N+1,\cdots,N+1;r),\]
		              there must be either $k-N+1$ states in $\{\ket{\psi_i}\}_{i=k-N+1}^{R_r(m)}$ which are orthogonal to each other with respect to $S_1$, or $\exists\ q\in\{2,\cdots,r\}$, s.t., there are $N+1$ states in $\{\ket{\psi_i}\}_{i=k-N+1}^{R_r(m)}$ which are orthogonal to each other with respect to $S_q$. If it is the first case, then by the inductive assumption we know the statement is true; if it is the latter case, WLG let $\{\ket{\psi_i^q}\}_{i=k-N+1}^{k+1}$ be an orthogonal set. Taking $M_{k-N+1,\cdots,k+1}^q$, if the outcome is not one of $k-N+1,\cdots,k+1$, then we exclude $N+1$ more states, that is, we can totally exclude $k-N-1+(N+1)=k$ states; if the outcome is one of $k-N+1,\cdots,k+1$, WLG let the outcome be $k+1$, then we can exclude $N$ states $\ket{\psi_{k-N+1}},\cdots,\ket{\psi_k}$, and if $\braket{\psi_{k+1}^q}{\psi_{k-N}^q}=0$, then exclude $\ket{\psi_{k-N}}$ and so we have already excluded $(k-N-1)+N+1=k$ states; if $\braket{\psi_{k+1}^q}{\psi_{k-N}^q}\neq0$, then there must be $v\in\{1,\cdots,r\}$($v\neq1,q$), s.t., $\braket{\psi_{k+1}^v}{\psi_{k-N}^v}=0$, then $M_{k+1,k-N}^v$ is sufficient to exclude at least one more state, then we can totally exclude $(k-N-1)+N+1=k$ states. Hence, we can exclude $(k-N-1)+(N+1)=k$ states totally.

		       \end{enumerate}
		     \end{enumerate}
		     By induction we can conclude that the result is true for all cases.
		     \end{enumerate}
	 	    \end{proof} 
	 	    
	 	    Notice that \autoref{Theorem1} is just the special form of \autoref{Theorem3} when $r=2$. From \autoref{Theorem3} we can obtain the following \autoref{Corollary2} by means of $r$-color Ramsey number, just similar to \autoref{Corollary1}. 
	 	    
	 	    \begin{corollary}\label{Corollary2}
	 	        (1) For any $R_r(m)$ orthogonal $r$-partite product states$(m\geq4, m\in\mathbb{N^{+}})$, a single copy is sufficient to exclude $m+1$ states via LOCC protocols; (2) For any $R(m,m,m)$ orthogonal tripartite product states $(m\geq5, m\in\mathbb{N^{+}})$, a single copy is sufficient to exclude $m+2$ states via LOCC protocols.
	 	    \end{corollary}
	 	    
	 	     \begin{proof}
	 	        We just need to check the relevant inequalities of \autoref{Theorem3}.
	 	         \begin{enumerate}[(1)]
	 	             \item First, we only need to check
	 	             \begin{equation*}
                      \begin{array}{l}
		                  R_r(m)-m \geq R(m+1,2,\cdots,2;r)=m+1,\\
		                  R_r(m)-(m+1) \geq R(m+2,1,\cdots,1;r)=1,\\
		              \end{array}
		              \end{equation*}
		              when $m\geq4$. From \autoref{Proposition4} we can conclude that $R_r(m)\geq R(m,m)$, and it's easy to check that $R(m,m)\geq 2m+1\geq m+2$ when $m\geq4$. So the two inequalities hold when $m\geq4$; 
		              \item Moreover, we  need to check that when $m\geq5$ we have
		              \begin{equation*}
		              \begin{array}{l}
                         R(m,m,m)-m\geq R(m+1,3,3),\\
                         R(m,m,m)-(m+1)\geq R(m+2,2,2),\\
                         R(m,m,m)-(m+2)\geq R(m+3,1,1).
                      \end{array}  
		              \end{equation*}
		              \begin{enumerate}[$\bullet$]
		                  \item The last two inequalities. We know $R(m,m,m)\geq R(m,m)$, and it's easy to see that $R(m,m)\geq2m+3$ when $m\geq5$. Thus $R(m,m,m)\geq2m+3\geq m+3$. That is, the last two inequalities hold when $m\geq5$;
		                  \item $R(m,m,m)-m\geq R(m+1,3,3)$. First from \autoref{Proposition5} and \autoref{Proposition6} we have $R(m,m,m)\geq3^{\frac{m}{2}}$ and $R(m+1,3,3)+m\leq2(m+1)^3+m$, and it's easy to show that $R(m+1,3,3)+m\leq2(m+1)^3+m<3^{\frac{m}{2}}<R(m,m,m)$ when $m\geq18$, that is, the inequality $R(m,m,m)-m\geq R(m+1,3,3)$ holds for $m\geq18$. When $5\leq m\leq 17$, from the known lower bound of $R(m,m,m)$, see \autoref{table4}, we can conclude that 
		                  \begin{equation*}
                             R(17,17,17)\geq\cdots\geq R(9,9,9)\geq14081\geq2\cdot(17+1)^3+17\geq R(17+1,3,3)+17.
                          \end{equation*}
                          That is, the inequality $R(m,m,m)-m\geq R(m+1,3,3)$ holds for $9\leq m\leq17$. And when $5\leq m\leq9$, comparing the known lower bound of $R(m,m,m)$(see \autoref{table4}) with the upper bound $2(m+1)^3+m$ of $R(m+1,3,3)$, we can also obtain that $R(m,m,m)-m\geq R(m+1,3,3)$ holds.
		              \end{enumerate}
	 	         \end{enumerate}
	 	    \end{proof}

	 	    \begin{table}[H]
	 	    \centering
	 	    \scalebox{1}{
\begin{tabular}{|c|c|c|c|c|c|c|c|}
\hline
\diagbox{r}{m}  & 3    & 4      & 5     & 6     & 7     & 8      & 9     \\ \hline
3 & 17   & 128    & 454   & 1106  & 3214  & 6132   & 14081 \\ \hline
4 & 51   & 634    & 4073  & 21302 & 84623 & 168002 &       \\ \hline
5 & 162  & 4176   & 38914 &       &       &        &       \\ \hline
6 & 538  & 32006  &       &       &       &        &       \\ \hline
7 & 1682 & 160024 &       &       &       &        &       \\ \hline
8 & 5288 &        &       &       &       &        &       \\ \hline
\end{tabular}
}
\caption{Known lower bounds of  $R_{r}(m)$\cite{radziszowski2011small}.}
\label{table4}
\end{table}

Finally, with the result above, we can prove the following \autoref{Theorem4}, which further reveals that $f_r(N)$ is exactly smaller than any form of $\lceil\epsilon N\rceil$($\epsilon>0$) once $N$ is sufficiently large. This indicates that $f_r(N)$ is far less than Shu's upper bound $\lceil\frac{N}{4}\rceil+1$ when $N$ is large\cite{shu2021locality}.  
	 	    
	 	   \begin{theorem}\label{Theorem4}
	 	        $\forall\ \epsilon>0$, $\exists N_0\in\mathbb{N }$, $s.t.$ $\forall\ N\geq  N_0$, $f_r(N) \leq \lceil\epsilon N\rceil$.
	 	    \end{theorem}
	 	    \begin{proof}
	 	         It is sufficient to prove that $\forall\ \epsilon>0$   $$\limsup_{N\rightarrow \infty}{\frac{f_r(N)}{N}}< \epsilon.$$ There exists a positive integer $m>4$, s.t., $\epsilon>\frac{1}{m}$.  For any orthogonal product states  with $N\geq R_r(m)$ elements, a single copy   is sufficient to exclude  $m+1$ states via  LOCC protocols. Written  $R_r(m)$ as the form 
            $$ R_r(m)=q_0(m+1)+r_0,  $$ where $q_0,r_0 \in \mathbb{N}$ and  $0\leq r_0<m+1.$   Then for any  $\forall\ N\geq (q_0+1)(m+1)$ which can be written as the unique form  $$ N=q_N(m+1)+r_N,  $$ where $q_N,r_N \in \mathbb{N},$    $q_N\geq q_0+1,$ and $0\leq r_N<m+1.$ By  \autoref{Corollary2},  we
can use $(q_N-q_0)$ copies to exclude $(m+1)(q_N-q_0)$ states and left $ q_0(m+1)+r_N$ states.   By Shu's result, we can always locally distinguish the remaining states provided 
$\lceil \frac{ q_0(m+1)+r_N} {4}\rceil+1 $  more copies.  Note that $\lceil \frac{ q_0(m+1)+r_N} {4}\rceil \leq \lceil \frac{ (q_0+1)(m+1)} {4}\rceil $. Therefore, to LOCC distinguish the $N$ product states, 
$(q_N-q_0)+\lceil \frac{ (q_0+1)(m+1)} {4}\rceil+1$ copies  are always  enough. That is, we always have 
$$f_r(N)\leq (q_N-q_0)+\lceil \frac{ (q_0+1)(m+1)} {4}\rceil+1. $$ Denote $M_\epsilon:= \lceil \frac{ (q_0+1)(m+1)} {4}\rceil+1-q_0$ which is a constant   depended on $\epsilon.$ Therefore, 
$$ \frac{f_r(N)}{N}\leq \frac{q_N}{N}+\frac{M_\epsilon}{N}\leq \frac{1}{m+1}+\frac{M_\epsilon}{N}.$$  Hence, we have  
 $$\limsup_{N\rightarrow \infty}{\frac{f_r(N)}{N}}\leq \limsup_{N\rightarrow \infty}\left( \frac{1}{m+1}+\frac{M_\epsilon}{N}\right)=\frac{1}{m+1}<\epsilon.$$

	 	    \end{proof}

		\section{Conclusion and Discussion}\label{sec:Conclusion} 
	In this paper, we studied the LOCC distinguishability of any $N$ orthogonal $r$-partite($r\geq2$) product states under the condition that multicopies of the states are allowed. By the language of Ramsey theory, we show that in any orthogonal set $\{\ket{\psi_i}=\ket{\psi_i^1}\otimes\cdots\otimes\ket{\psi_i^r}\}_{i=1}^{N}$, no matter what the structure of $\{\ket{\psi_i}=\ket{\psi_i^1}\otimes\cdots\otimes\ket{\psi_i^r}\}_{i=1}^{N}$ is, there must be some states that are orthogonal to each other with respect to a subsystem, once $N$ is large enough.
		
		From the observation above, we gave a sufficient condition for any $R_r(m)$ orthogonal $r$-partite product states to exclude $k$ states via a single copy, and as a corollary we proved that for any $R_r(m)$ $r$-partite product states($r\geq2, m\geq4$), a single copy is sufficient to exclude $m+c$($c=1,2,3$) states via LOCC protocols. Furthermore, we showed that the least number $f_r(N)$ of copies  which ensures the LOCC distinguishability of any $N$ orthogonal $r$-partite product states is smaller than $\lceil\epsilon N\rceil$ for any $\epsilon>0$ , once $N$ is sufficiently large, rather than just $f_r(N)\leq\lceil\frac{N}{4}\rceil+1$, thus improving Shu's result \cite{shu2021locality}. And particularly for bipartite case, we give a new upper bound $\lceil\frac{N}{6}\rceil+2$ of $f_2(N)$, which is less than Shu's upper bound $\lceil\frac{N}{4}\rceil$ when $N>24$.
		
		We used the language of Ramsey theory to formalize the existence of some local substructure of any sufficiently large orthogonal set of product states.  Thus we  obtained  better results for the LOCC distinguishability of any orthogonal set of product states via multicopies. The usage of Ramsey theory shows a formal way to deal with the localization problem of product states, which relies on the results of Ramsey theory. And we believe that the viewpoint of Ramsey theory can help us improve the understanding of locality of product states.

	\end{document}